\documentclass[11pt]{article}

\usepackage[letterpaper, margin=1in]{geometry}
\usepackage[english]{babel}
\usepackage[utf8]{inputenc}
\usepackage{amsmath,amssymb,amsthm}
\usepackage{graphicx}
\usepackage[colorinlistoftodos]{todonotes}
\usepackage{comment}

\usepackage{tikz}
\usetikzlibrary{trees}
\usepackage{graphicx}
\usetikzlibrary{positioning, quotes}
\usepackage{caption}
\usepackage{subcaption}
\usepackage{cancel}
\usepackage{tikz}
\usepackage{comment}
\usepackage{paralist}
\usepackage{wrapfig}

\usepackage[colorlinks=true,urlcolor=blue, citecolor=red,linkcolor=blue,linktocpage,pdfpagelabels, bookmarksnumbered,bookmarksopen]{hyperref}
\usepackage[hyperpageref]{backref}
\usepackage[noabbrev]{cleveref}
\usepackage[inkscapelatex=false]{svg}

\usepackage[colorlinks=true,urlcolor=blue, citecolor=red,linkcolor=blue,linktocpage,pdfpagelabels, bookmarksnumbered,bookmarksopen]{hyperref}
\usepackage[hyperpageref]{backref}
\usepackage[noabbrev]{cleveref}

\newtheorem{theorem}{Theorem}
\newtheorem{lemma}[theorem]{Lemma}

\newtheorem{observation}[theorem]{Observation}

\theoremstyle{definition}

\newtheorem{conjecture}[theorem]{Conjecture}

\begin{document}

\title{Competitive Online Transportation Simplified}
\author{Stephen Arndt\thanks{Tepper School of Business, Carnegie Mellon University.} \and Benjamin Moseley\thanks{Tepper School of Business, Carnegie Mellon University. Supported in part by Google Research Award, NSF grants CCF-2121744 and CCF-1845146 and ONR Grant N000142212702.} \and Kirk Pruhs\thanks{Computer Science Department, University of Pittsburgh. Supported by National Science Foundation grant CCF-2209654.} \and Marc Uetz\thanks{Faculty of Electrical Engineering, Mathematics, and Computer Science, University of Twente. Marc Uetz acknowledges funding by the faculty EEMCS of the University of Twente, and the Tepper School of Business of CMU for hosting as a short-term scholar.}}
\date{}

\maketitle

\begin{abstract}
The setting for the online transportation problem is a metric space $M$, populated by $m$ parking garages of varying capacities. Over time cars arrive in $M$, and must be irrevocably assigned to a parking garage upon arrival in a way that respects the
garage capacities. The objective is to minimize the aggregate distance traveled by  the cars. 
In 1998, Kalyanasundaram and Pruhs conjectured that there is a $(2m-1)$-competitive deterministic algorithm for the online transportation problem, matching the optimal
competitive ratio for the simpler online metric matching problem. 
Recently, Harada and Itoh presented the first $O(m)$-competitive deterministic algorithm for the online transportation problem. 
Our contribution is an alternative algorithm design and analysis that 
we believe is simpler. 
\end{abstract}

\newpage

\section{Introduction}

Let us start by defining the problem.

\paragraph{Online Transportation Problem Definition:}
The setting for the online transportation problem is a metric space $M=(X, d)$, consisting of a collection $X$ of points/locations,
and  a distance function $d: X^2 \rightarrow 	\mathbb{R}^+$ that specifies the distance
between pairs of points in $X$. The metric space $M$ is populated by $m$ parking garages $1, 2, \dots, m$ at the collection of locations $G=\{g_1, g_2, \ldots, g_m\}$, where each $g_i \in X$.
Further each garage $i$ has a 
positive  integer capacity $c_i$, which is conceptually 
the number of parking spaces in  garage $i$.
At each integer unit of time $t$, $1 \le t \le k $, a car $t$ arrives at some location $x_t \in X$, where $k=\sum_{i=1}^m c_i$.
After the online algorithm sees the location $x_t$ it must irrevocably assign car $t$ 
to a garage $i$ of its choice, subject to the constraint that strictly less than $c_i$ cars have been previously assigned to garage $i$,
incurring a movement cost of $d(x_t, g_i)$. The objective is to  minimize the total movement (cost) of the cars.

\subsection{The Story to Date}

A widely-studied special case of the online transportation problem is the
online metric matching problem. In the online metric matching problem,
all of the garages have unit capacity. 
At the dawn of the halcyon era of competitive analysis of online algorithms, circa 1990, 
it was shown that:
\begin{itemize}
    \item The natural greedy algorithm, which parks each car in the
nearest available garage, is $\left(2^m - 1\right)$-competitive \cite{onlineweightedmatching}. 
\item  The Retrospective/Permutation algorithm is $(2m-1)$-competitive \cite{onlineweightedmatching,khullermitchell}. 
The Retrospective algorithm maintains the invariant that residual capacity of each garage is 
the same as residual capacity for that garage in the optimal matching of the cars that have arrived to date. 
\item  The competitive ratio of every deterministic algorithm on a star metric space is
at least $2m-1$ \cite{onlineweightedmatching, khullermitchell}. So for a general metric space, the Retrospective algorithm is optimally competitive among deterministic algorithms. 
\end{itemize}

In the mid 1990's, 
one of the three natural followup research directions/questions, 
proposed in \cite{Kalyanasundaram1998}, was to determine the optimal deterministic competitive ratio
for the online transportation problem.  In fact it was conjectured that:

\begin{conjecture}\cite{Kalyanasundaram1998}
\label{conj:2m-1}
There is a $(2m-1)$-competitive deterministic algorithm for the online transportation problem.     
\end{conjecture}

So essentially this conjecture hypothesizes that arbitrary-capacity garages are no harder
for an online algorithm than unit-capacity garages.
It was observed early on that the analysis of the natural greedy algorithm for the online metric
matching algorithm is also valid in the online transportation problem, so the natural greedy algorithm is $\left(2^m-1\right)$-competitive for online transportation~\cite{Kalyanasundaram1998}. In contrast, the competitive ratio of the Retrospective algorithm 
for online transportation is not bounded by any function of $m$~\cite{Kalyanasundaram1998}.

There are two hurdles that researchers faced when trying  to find a (nearly) optimally-competitive
algorithm for online transportation. 
The first hurdle is  an algorithm design issue. 
In particular, it is not clear how the garage capacities should affect the design of the algorithm. 
One natural hypothesis is that the ``right'' algorithm should try to load balance
between garages,  favoring sparsely-filled garages over nearly-full garages.
In 1995, this intuition led to a result showing that a modification of the natural greedy algorithm, where the distance to garages that are more than half full is artificially
doubled, is $O(1)$-competitive for online transportation with resource augmentation, that
is when comparing to the optimal objective value with garages of half the capacity~\cite{onlinetransportwa}.
But this approach essentially hit a dead end, as such algorithms are not readily adaptable
to be close to  $O(m)$-competitive for online transportation without resource augmentation. 
The second, seemingly higher,  hurdle is an algorithmic analysis issue.
In the $(2m-1)$-competitiveness analysis of the  Retrospective algorithm, the distance that the Retrospective algorithm
moves each car is upper-bounded by twice the cost of the (entire) optimal matching. 
In the online transportation problem, due to there being potentially significantly many more cars than garages,
the distance that the online algorithm moves a car has to be upper-bounded by only a portion, say 
the cost of an edge for example, of the optimal matching. Identifying that portion/edge to charge
is challenging for the known algorithms that are candidates for being $O(m)$-competitive
for online transportation. 
 
As an illustrative example of this analysis difficulty, consider the 
Robust Matching (RM) algorithm~\cite{raghvendra2016}. 
Roughly speaking the RM algorithm tries to simultaneously approximate both the 
natural greedy algorithm and the Retrospective
algorithm, by both always parking each car in a garage that is within a constant factor
of being the closest, and by 
maintaining the invariant that the residual capacity of each garage is 
the same as the residual capacity for that garage in some constant-approximate optimal matching of the cars that have arrived to date~\cite{raghvendra2016}. In particular, the RM algorithm does not, in any overt way, take the fullness of the
various garages into account. 
The name Robust Matching was given to this algorithm when it
was introduced into the literature in 2016~\cite{raghvendra2016}, although the algorithm, and
the closely-related generalized work function algorithm, were known
as reasonable candidate algorithms in the mid 1990's. 
While it seems likely (at least to us) that  the RM algorithm is  $O(m)$-competitive for the online transportation problem,
the best competitiveness anyone can show is 
$O(m \log^2 k)$~\cite{HaradaI25,NayyarR17}. The challenge of establishing $O(m)$-competitiveness for RM
is determining how to charge RM's movement of individual cars to portions/edges in the optimal matching. 

Finally, after 30 years,  in 2025, a deterministic 
$O(m)$-competitiveness result for online transportation
was obtained~\cite{HaradaI25}. 
At a high level  this result surmounted the analysis issue by designing an
algorithm, which we will call the Star Decomposition algorithm,
with the Most Preferred Free Servers (MPFS) property,
which means that the garage that the algorithm picks to park a car only
depends on the location of the garages and the  car's arrival
location (not on how full the garages are, not on where previous cars were parked, etc.). 
It was observed in \cite{HaradaIM24} that,
like the natural greedy algorithm, upper bounds on the competitiveness
of algorithms with the MPFS property for the online metric matching problem
carry over to the online transportation problem. 
That is, if an MPFS algorithm $A$ is $f(m)$-competitive for online metric matching
then $A$ is $f(m)$-competitive for online transportation;
this fact is a natural consequence of
Vizing's Theorem applied to bipartite graphs. So Vizing's Theorem seemingly  automatically handles how to identify 
the edges in the optimal matching to charge the algorithm's movement of a car, if the algorithm has the MPFS property.

The Star Decomposition algorithm a priori computes:
\begin{itemize}
    \item A minimum spanning tree $S$ of the metric space $M_G$, which is the original metric
space $M$ restricted to the garage locations.
\item The tree $T$ obtained from $S$ by rounding distances up to the next integer power of two (this explains
the first factor-of-two loss in the computed competitive ratio, relative to a desired $2m-1$ bound).
\item A recursive star decomposition of $T$ in a hierarchical 
collection of star metrics (somewhat reminiscent of the decomposition of an  arbitrary
metric into a hierarchically separated tree). 
\item
From this star decomposition, for garage $i$,  a permutation $\pi_i$ of the garages. 

\end{itemize}
Then when a car $x_t$ arrives, the Star Decomposition algorithm pretends
that the  car arrived at the location
of the garage $i$ closest to the actual arrival location of the car (this explains
the second factor-of-two loss in the computed competitive ratio, relative to a desired $2m-1$ bound).
The Star Decomposition algorithm then parks the car in the first   unfull garage 
in the permutation $\pi_i$.

We now turn to the analysis of the Star Decomposition algorithm in \cite{HaradaI25}.
It is  
straightforward to observe that the
algorithm has the MPFS property. 
 The main issue/hurdle in the analysis is showing that the optimal matching in 
 the original metric space $M$ is relatively expensive, 
 or equivalently that the adversary can't benefit from using 
edges in $M$ that shortcut $T$.
 Intuitively this hurdle exists because the Star Decomposition algorithm
 does not overtly keep track of any sort of approximation of the current optimal matching,
 as does say the RM algorithm. 
But eventually, after significant work, \cite{HaradaI25} is able to show,
using the hybrid technique borrowed from \cite{harmonic-alg}, 
that the Star Decomposition algorithm is
$(8m-5)$-competitive on the original metric space $M$.
So at a high level the analysis of the Star Decomposition algorithm in \cite{HaradaI25}
avoided the issue of identifying an edge in the optimal matching to charge a car's movement to 
(because the algorithm has the  MPFS property) at the expense of complicating the 
argument that the optimal movement cost is large.

\subsection{Our Contribution}

Our goal here is to simplify the algorithm design
and analysis in \cite{HaradaI25}, so that it is easier to build on the significant insights there
toward the larger goal of settling Conjecture \ref{conj:2m-1}.

In Section \ref{sec:icp}, our first step  towards accomplishing this simplification is to
borrow a technique from \cite{AntoniadisBNPS19}, and
consider a relaxation of the
online transportation problem, which we call the itinerant car problem,
in which we remove the restriction that cars have to irrevocably assign
to a garage at arrival.

\paragraph{Itinerant Car Problem Definition:}
The setting and input is the same as in the online transportation problem. 
However, after the online algorithm sees the location $x_t$ that car $t$ arrives,
the algorithm can move each of the $t$ cars
that have arrived up until time $t$, from their current garage, to an arbitrary other  garage, 
 subject to the constraint that afterwards each garage $i$ can contain at most $c_i$ cars. 
The objective is still to  minimize the total movement (cost) of the cars.

Note that any feasible solution to an online transportation problem instance is a feasible solution
to the corresponding itinerant car problem instance, and an optimal solution to an online transportation problem 
instance is still an optimal solution to the corresponding itinerant car problem instance.

We then define a relatively simple online algorithm $\mathcal A$ for the itinerant car problem.
Our algorithm $\mathcal A$ is defined on the same power-of-two tree metric $T$,
derived from the general metric space $M$, as the Star Decomposition algorithm (so the garages are
the vertices of $T$, and all edge weights in $T$, which correspond to rounded distances in $M$, are integer powers of two).
In $\mathcal A$ each car $t$ independently drives around $T$,
starting from the garage $i$  where it arrived, looking for a garage that is not full.
The route driven by car $t$ arriving at garage $i$ consists of a concatenation of 
a sequence of depth-first-search walks $W_0(i), W_1(i), W_2(i), \ldots$, where the level $j$ 
walk $W_j(i)$ visits the garages reachable from garage $i$ without
traversing any edges with weight $2^j$ or greater. The algorithm $\mathcal A$ maintains the invariant
that garages always contain the cars that have visited them on the lowest-level walk.

The definition of our algorithm $\mathcal A$ is  similar to the intuitive description of the
Star Decomposition algorithm in \cite{HaradaI25}. 
The two most notable differences are 
\begin{itemize}
    \item The walks in \cite{HaradaI25} are coordinated in the order that visit various portions of the metric space (and this fact is used in the analysis). 
    \item Instead in algorithm $\mathcal A$ 
cars on lower-level walks are given priority, which isn't 
the case in \cite{HaradaI25}, and this alleviates the need for the walks to be coordinated. 
\end{itemize}
We speculate that the formal definition of the Star Decomposition
algorithm in \cite{HaradaI25} is more complicated than is minimally needed  because this
additional complexity in the algorithm description simplifies some aspects of the 
subsequent algorithm analysis. 

The key insight in our analysis of algorithm $\mathcal A$ is to note that that the final parking spots of the cars in algorithm $\mathcal{A}$ essentially matches the optimal bottleneck matching. An optimal bottleneck matching of cars to parking spots in a tree $T$ is a min-cost matching, where the cost of assigning a car at vertex $u$ to a parking spot at vertex $v$ is the weight of the largest edge on the $u-v$ path in $T$. This observation is a relatively straightforward consequence of the properties of
the algorithm. We then extend algorithm $\mathcal A$ for a power-of-two tree metric $T$
to an algorithm $\mathcal B$ for a general metric space $M$ 
as in \cite{HaradaI25}. 
We conclude Section \ref{sec:icp} by noting that a straightforward
consequence of this observation  is that
 algorithm $\mathcal B$ is $(8m-7)$-competitive for the itinerant car problem.

Then in Section \ref{sec:trans} we turn to the online transportation problem.
We first define an algorithm $\mathcal C$, derived from the itinerant car algorithm $\mathcal B$.
In particular, the algorithm $\mathcal C$ maintains the invariant that the residual capacity 
of every garage for $\mathcal C$ is the same as the residual capacity of that garage for algorithm $\mathcal B$.
We then observe that by the triangle inequality that the movement cost for $\mathcal C$ can
be at most the movement cost for $\mathcal B$. 
From this we can then immediately conclude that algorithm $\mathcal C$ is $(8m-7)$-competitive
for the online transportation problem. 

So at a high level, our use of the relaxed itinerant car problem 
served the same ends as the use of the  hybrid technique~\cite{harmonic-alg}
in \cite{HaradaI25}, namely establishing a lower bound to the optimal cost,
however our relaxation allows this to be accomplished much more simply/cleanly. 
We also show in Section \ref{sec:trans}  that our algorithm $\mathcal C$ does not have
the MPFS property, so our use of the relaxed itinerant car problem
also alleviated the algorithmic design need for the MPFS property. 
We think it is at least plausible that the use of such a relaxation may be useful
in resolving other open problems related to online metric  matching. 

\subsection{Other Related Results}

Besides generalizing online metric matching results to online transportation, the two other main
open questions posed in \cite{Kalyanasundaram1998} were to determine the optimal randomized competitive ratio
against an oblivious adversary, and to find optimally-competitive algorithms for particular
metric spaces of interest, most notably the line.

For every metric space, the competitiveness of the RM algorithm is
within a $O(\log^2 m)$ factor of the optimal competitive ratio
achievable by a deterministic online algorithm
for the online metric matching problem~\cite{NayyarR17}.
For the line metric, the RM algorithm is $O(\log m)$-competitive
for the online metric matching problem~\cite{Raghvendra18}, and the competitive ratio of every randomized algorithm against an oblivious
adversary is $\Omega(\log^{1/2} m)$~\cite{PesericoS23}. 
There are several other more specialized results
for online metric matching on a line in the literature~\cite{AntoniadisBNPS19,AntoniadisFT18,harmonic-alg,HaradaI23,elias-OML}.

There is a lower bound of $\Omega(\log m)$ on the randomized
competitive ratio for online metric matching for a general metric 
space. The best known upper bound on the randomized competitive
ratio for online metric matching on a general metric space
is $O(\log^2 m)$~\cite{random-O(log2k),meyerson}. 
This upper bound was extended to online transportation in \cite{KalyanasundaramPS23}, improving on an earlier result that required additive resource augmentation~\cite{ChungPU08}.

The  natural Greedy algorithm was recently shown to be
$O(1)$-competitive for online transportation with a factor-of-three resource augmentation~\cite{ArndtAP23}.

\section{Itinerant Car Problem}\label{sec:icp}

In this section, we present an $(8m-7)$-competitive algorithm for the itinerant car problem. In \Cref{subsec:alg_A}, we give a complete description of algorithm $\mathcal{A}$ for a power-of-two tree metric $T = (V, E)$. In \Cref{subsec:analysis_A}, we analyze algorithm $\mathcal{A}$, and prove it is $(2m-2)$-competitive for the itinerant car problem on power-of-two tree metrics. In \Cref{subsec:alg_B}, we describe algorithm $\mathcal{B}$ for general metric spaces, and prove it is $(8m-7)$-competitive for the itinerant car problem, in a way similar to \cite{HaradaI25}.

\subsection{Algorithm $\mathcal{A}$ for a Power-of-Two Tree Metric}\label{subsec:alg_A}

We start with some definitions. 
Let $T = (V, E)$ be a power-of-two tree metric. 
That is,  one garage is located at each vertex in $V$,
each edge in $E$ has an associated weight that is of the form $2^j$ for some integer $j \in [0, n]$,
and the distance between any two points is the length of the unique path between these points
in the tree $T$. Further, cars will arrive only at vertices of $T$. A garage is \textbf{full} if the number of cars it contains equals its capacity, and \textbf{unfull} otherwise. Define $T_j(i)$ to be the level $j$ tree for garage $i \in [m]$,
that is the  portion of $T$ reachable from garage $i$ without
traversing any edges with weight  $2^j$ or greater. Note $T_0(i) = i$ and $T_{n+1}(i) = T$ for all $i \in [m]$. Define $W_j(i)$ 
as an arbitrary depth-first-search walk/tour on $T_j(i)$, starting from $i$,
which we will call the level $j$ walk for garage $i$.
Alternatively 
$W_j(i)$ can be thought of as an Eulerian tour that traverses each edge in $T_j(i)$ twice,
once in each direction. 
We now give an algorithm for the itinerant car problem on 
a power-of-two tree metric. 

\paragraph{Algorithm $\mathcal A$ Definition:}
Each car $t$ that arrives at garage $i$ first drives along the route of the walk $W_0(i)$ (which is just the vertex at garage $i$), then along the route of the walk $W_1(i)$, then along the route of $W_2(i)$, etc.
If while traversing $W_j(i)$ this car $t$ arrives at a garage $h$, the following occurs:
\begin{itemize}
    \item If garage $h$ is currently unfull, then car $t$ stops at garage $h$ and parks there. 
    \item If the garage $h$ is full, then let car $t'$ be a car currently parked in garage $h$ currently
    on the highest walk level. Say car $t'$ arrived at garage $a \in [m]$, and is currently
    on its level $b$ walk $W_b(a)$ (so no car currently parked in $h$ is on its level $c$ walk, where $c > b$).
    \begin{itemize}
        \item  If $j < b$, then car $t$ parks in garage $h$, ejects car $t'$, and car $t'$ continues traversing $W_b(a)$.
        \item  Otherwise, if $j \geq b$, then car $t'$ stays parked in garage $h$, and car $t$ continues traversing $W_j(i)$.
    \end{itemize}
\end{itemize}
 So intuitively the car that started its walk closer to garage $h$ (using the bottleneck distance metric) ejects the car that started its walk furthest away from garage $h$.

\subsection{Analysis of Algorithm $\mathcal{A}$}\label{subsec:analysis_A}

In this section, we prove \Cref{thm:icp_A}, which states that algorithm $\mathcal{A}$ is $(2m-2)$-competitive for the itinerant car problem on power-of-two tree metrics. The main technical lemma is \Cref{lemma:inc_count}, which establishes that the final walk levels of the cars in algorithm $\mathcal{A}$ essentially matches the optimal bottleneck matching.



We start by introducing some notation. Let $I$ be an arbitrary instance of
the itinerant car problem on the metric space $M$. 
We use ${\mathcal B}(I)$ to denote the cost of algorithm $\mathcal B$ on input $I$ and
${\mathcal A}(J)$ to denote the cost of algorithm $\mathcal A$ on the input $J$ given to 
it by algorithm $\mathcal B$. 
We use $\mbox{Opt}(Z, Y)$ to denote the optimal matching cost for instance $Z$ in metric space $Y$. 
If $Y$ is a tree metric, we 
define the bottleneck cost for parking a car arriving at garage $i$ in garage $j$
as the weight of the 
maximum-weight edge on the path between  $i, j$ in the tree $Y$,
and the bottleneck cost of a matching as the aggregate bottleneck cost of
its car-garage pairings. 
If $Y$ is a tree metric, we use $\mbox{BottleOpt}(Z, Y)$ to denote the optimal bottleneck
matching cost for instance $Z$ in metric space $Y$. 
For a car $t$, let $\ell(t)$ be the highest level of a walk that it reached
using algorithm $\mathcal A$ on instance $J$. 

\begin{observation}\label{obs:inc}
In algorithm $\mathcal{A}$, if a car that arrived at garage $i$ has begun its level $\ell$ walk $W_{\ell}(i)$, then it must be the case that all the garages in $T_{\ell-1}(i)$ are currently full, and all cars parked in $T_{\ell-1}(i)$ are on walks with level at most $\ell-1$.
\end{observation}

\begin{lemma}\label{lemma:inc_count}
$\sum_{t=1}^k 2^{\ell(t)-1}  = \mbox{BottleOpt}(J, T)$.
\end{lemma}

\begin{proof}
We prove  that the number of cars $t$ with $\ell(t) \ge j+1$ equals the number of cars  $t$ with bottleneck matching cost  $2^j$ or greater in an optimal bottleneck matching.
Let $C_1, \ldots C_h$ be the connected components of $T$ when the edges $B$ of
weight $2^j$ or greater are removed. 
Consider an arbitrary such connected component $C_i$. 
 Let $s_i$ the number of cars that arrive in $C_i$ minus the aggregate capacity of
 the garages in $C_i$ if this is positive, and zero otherwise. 
Then the number of cars $t$ that algorithm $\mathcal A$ moves across  the edges 
in $B$, and thus have $\ell(t) \geq j+1$,  is $\sum_{i=1}^h s_i$, since algorithm $\mathcal A$ gives priority to cars on lower-level walks.
Further the number of edges of weight $2^j$ or greater in an optimal bottleneck matching 
is also  $\sum_{i=1}^h s_i$, as the natural greedy algorithm computes an optimal
bottleneck matching (as can be shown by a simple exchange argument). 
\end{proof}

\begin{theorem}\label{thm:icp_A}
 $\mathcal{A}(J) \leq (2m-2)\cdot \mbox{BottleOpt}(J, T)$.
\end{theorem}

\begin{proof}
For a car $t$ arriving at a garage $i$, 
let $e_j(t)$ be the number of edges with weight $2^j$ in $T_{\ell(t)}(i)$. 
The claim then follows from the following sequence of inequalities:
\begin{align}
{\mathcal A}(J) & \le \sum_{t=1}^k \sum_{j=0}^{\ell(t)-1} 2(\ell(t)-j) \cdot 2^j e_j(t) \label{eq:fubar1} \\
&\leq \sum_{t=1}^k \sum_{j=0}^{\ell(t)-1} 2^{\ell(t)-j} \cdot 2^j e_j(t) \label{eq:fubar2}\\
  &=  \sum_{t=1}^k  \sum_{j=0}^{\ell(t)-1} 2^{\ell(t)} e_j(t) \label{eq:fubar3}\\
  &=\sum_{t=1}^k  2^{\ell(t)} \sum_{j=0}^{\ell(t)-1}  e_j(t) \label{eq:fubar4}\\
   &\leq\sum_{t=1}^k  2^{\ell(t)} (m-1) \label{eq:fubar5}\\
  &= (2m-2) \cdot \sum_{t=1}^k 2^{\ell(t)-1}  \label{eq:fubar6}\\
  &= (2m-2)\cdot \mbox{BottleOpt}(J, T) \label{eq:fubar7}
\end{align}
Line (\ref{eq:fubar1}) follows since car $t$ arriving at garage $i$
will traverse the edges in $T_{j+1}(i)$ at most $2(\ell(t)-j)$ times for $j \leq \ell(t)-1$,
and an edge of weight $2^j$ will only appear in trees 
$T_{j+1}(i), T_{j+2}(i), \dots, T_{\ell(t)}(i)$. 
Line (\ref{eq:fubar5}) follows since tree $T$ has $m-1$ edges. 
Line (\ref{eq:fubar7}) follows from Lemma
\ref{lemma:inc_count}.
\end{proof}

\subsection{Algorithm $\mathcal{B}$ and Analysis for General Metric Spaces}\label{subsec:alg_B}

We now describe an algorithm $\mathcal{B}$ for the itinerant car problem on general metric spaces, and prove in \Cref{thm:icp_B} that algorithm $\mathcal{B}$ is $(8m-7)$-competitive. We assume without loss of generality that no
two garages are in the same location, and scale distances so that the minimum distance between two garages is 1.

\paragraph{Algorithm $\mathcal{B}$ Definition:} Algorithm $\mathcal{B}$ simulates algorithm $\mathcal{A}$. Algorithm $\mathcal{B}$ first computes a minimum spanning tree $S$ of $M_G$, where $M_G$ is the original metric space $M$ restricted to the garage locations. Algorithm $\mathcal{B}$ then computes a tree $T = (V, E)$ from $S$ by rounding distances up to the next integer power of two. Algorithm $\mathcal{B}$ simulates algorithm $\mathcal{A}$ on $T = (V, E)$ in the following way: for each car $t$ arriving at location $x_t \in X$, algorithm $\mathcal{B}$ moves car $t$ to the nearest garage $i$ to $x_t$, simulates algorithm $\mathcal{A}$ on $T$ as if car $t$ arrived at garage $i$, and moves the cars identically to how algorithm $\mathcal{A}$ moves its cars on $T$.

\begin{lemma}
\label{lem:bottle4}
$\mbox{BottleOpt}(J, T)  \le 4 \cdot \mbox{Opt}(I, M)$.   
\end{lemma}

\begin{proof}
\begin{align}
\mbox{BottleOpt}(J, T) &\leq 2 \cdot \mbox{BottleOpt}(J, S) \label{eq:bubba1}\\
&\leq  2  \cdot \mbox{Opt}(J, M) \label{eq:bubba2} \\
&\leq 4 \cdot \mbox{Opt}(I, M) \label{eq:bubba3}
\end{align}
Inequality (\ref{eq:bubba1}) follows because $T$ is constructed from 
the minimum spanning tree $S$ by rounding weights up by at most a factor of two. 
Inequality (\ref{eq:bubba2}) follows from the cycle property of minimum spanning trees,
that every edge $(i, j)$ in $M_G - S $  must have weight at least as large as any 
edge on the path between $i$ and $j$ in $S$. 
Inequality (\ref{eq:bubba3}) follows from the fact that the aggregate distance from
the arrival location of each car in $M$ to the nearest garage is a lower bound
to the optimal cost  $\mbox{Opt}(I, M)$, and the triangle inequality.
\end{proof}

\begin{theorem}\label{thm:icp_B}
Algorithm $\mathcal{B}$ is $(8m-7)$-competitive for the itinerant car problem.
\end{theorem}

\begin{proof}
The statement then follows from the following sequence of inequalities:
\begin{align}
{\mathcal B}(I) &\leq {\mathcal A}(J) +  \mbox{Opt}(I, M) \label{eq:dubba1}\\
&\leq (2m-2) \cdot \mbox{BottleOpt}(J, T)   +  \mbox{Opt}(I, M) \label{eq:dubba2} \\
&\leq 4 (2m-2)  \cdot \mbox{Opt}(I, M)  + \mbox{Opt}(I, M)   \label{eq:dubba3}\\
&= (8m-7) \cdot \mbox{Opt}(I, M) 
\end{align}
Inequality (\ref{eq:dubba1}) follows from the fact that the aggregate distance from
the arrival location of each car in $M$ to the nearest garage is a lower bound
to the optimal cost  $\mbox{Opt}(I, M)$, and the fact that the distance 
between two locations in $T$ is an upper bound for the distance between these locations in $M$.
Inequality (\ref{eq:dubba2}) follows from Theorem \ref{thm:icp_A}.
Inequality (\ref{eq:dubba3}) follows from Lemma \ref{lem:bottle4}.
\end{proof}

\section{Online Transportation Problem}
\label{sec:trans}

Here we give an $(8m-7)$-competitive algorithm $\mathcal C$ for the online transportation problem that is based on our algorithm $\mathcal B$ for the itinerant car problem. 

\paragraph{Algorithm $\mathcal C$ Description:} 
Algorithm $\mathcal C$ simulates the itinerant cars algorithm 
$\mathcal B$. Algorithm $\mathcal C$ assigns each car $t$ 
to the unique garage $j$ with the property that algorithm 
$\mathcal B$ has one more  car parked in garage $j$
after car $t$ arrives than before car $t$ arrived.

\begin{theorem}\label{thm:otp_main}
Algorithm $\mathcal C$ is well-defined, and 
is $(8m-7)$-competitive for the online transportation problem.
\end{theorem}

\begin{proof}
The response of algorithm $\mathcal B$ to the arrival of a car $t$ at a location $x_t$ can be described as a sequence of car movements/walks
$W_1, \ldots W_h$, where 
\begin{itemize}
    \item Walk $W_1$ is the movement of car $\beta(1) = t$
to some garage $\alpha(1)$. 
\item The start of walk $W_i$ is garage $\alpha(i-1)$, $2 \le i \le h$. 
\item The car $\beta(i)$ moving on walk $W_i$ is the car ejected from garage
$\alpha(i-1)$ when car $\beta(i-1)$ terminated its walk there, $2 \le i \le h$. 
\item The garage $\alpha(h)$ was unfull right before car $t$ arrives. 
\end{itemize}
The fact $\mathcal C$ is well-defined follows immediately from 
this understanding of algorithm $\mathcal B$. 
The movement cost of algorithm $\mathcal C$ is at most the movement
cost of algorithm $\mathcal B$ by the triangle inequality.
The competitive ratio of algorithm $\mathcal C$ for the online
transportation problem is at most the competitive
ratio of algorithm $\mathcal B$ for the itinerant car problem
because an optimal solution for the online transportation problem
is also optimal for the itinerant car problem. The claim follows from \Cref{thm:icp_B}.
\end{proof}

\begin{wrapfigure}{r}{0.6\textwidth}
\centering
\includegraphics[width=0.58\textwidth]{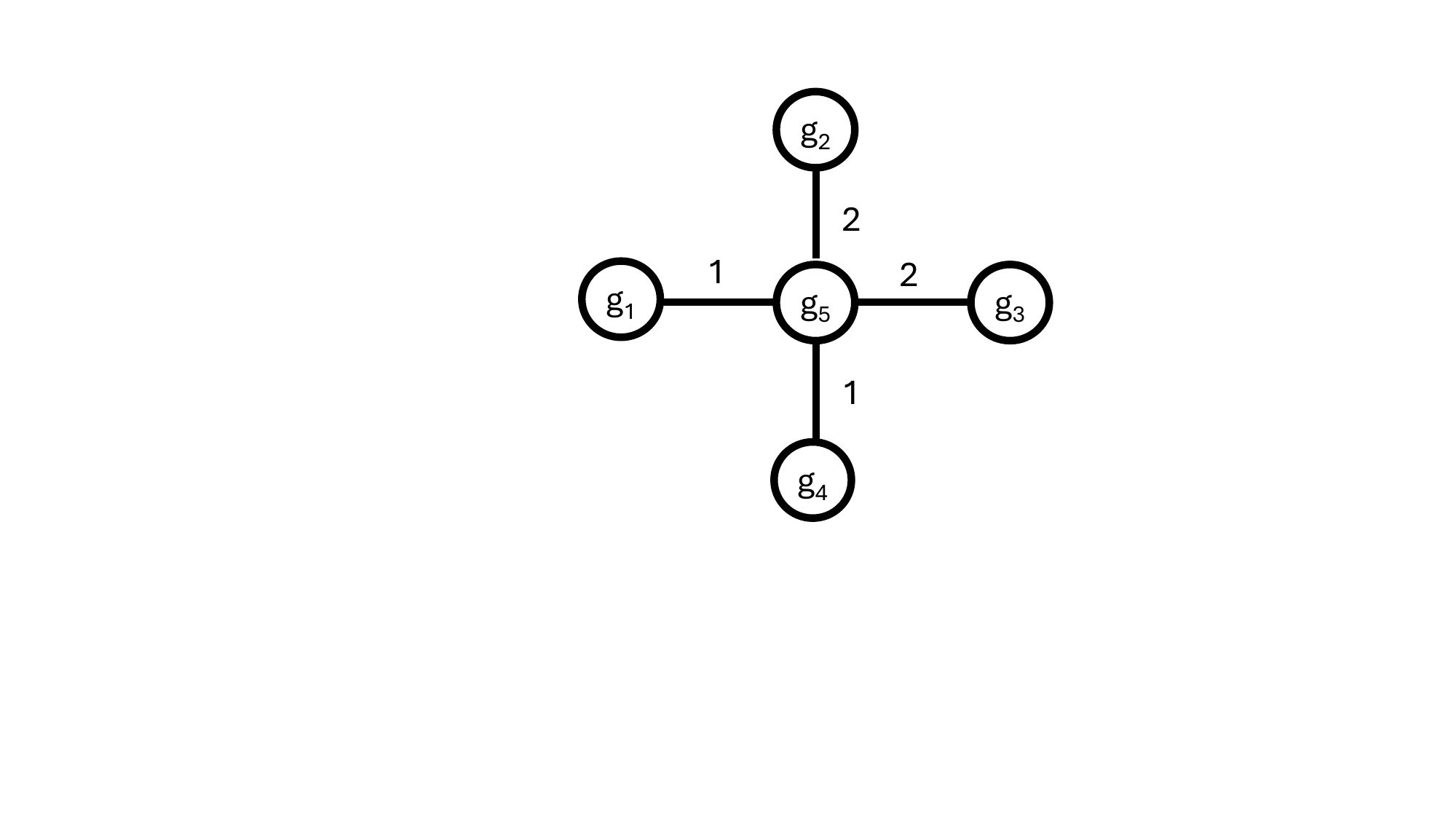}
\caption{Example tree metric $T$ showing that algorithm $\mathcal C$ does not have the MPFS property}
\label{fig:notMPFS}
\end{wrapfigure}
Lastly, we briefly explain why algorithm $\mathcal C$ does not have
the MPFS property. Consider the metric space in Figure \ref{fig:notMPFS},
that is already a power-of-two tree metric $T$, where all garages have
unit capacity. 
Consider a car $t$ that arrives at the garage at location $g_1$,
and that must search for another garage because this garage is filled
by an earlier car that arrived there. Then car $t$ would next
move to the garage at $g_5$. If the garage at $g_5$ was
filled with a car arriving earlier at $g_5$, then car $t$ would
have to move to the garage at $g_4$. However, if
the garage at $g_5$ was
filled with a car arriving earlier at $g_3$, then car $t$ could
next move to the garage at $g_2$ (because in algorithm $\mathcal A$, the
walk followed by cars arriving at $g_3$ could visit $g_2$ before $g_4$). 
Thus the location that car $t$ parks is not solely determined by the
arrival location of the car and the location of the unfull garages.

\clearpage

\bibliography{main}
\bibliographystyle{plain}

\end{document}